\documentclass[submission]{eptcs}

\usepackage{amsfonts,latexsym, amssymb, amsmath}
\usepackage{stmaryrd}
\usepackage{proof, bussproofs}
\usepackage{pb-diagram}

\newtheorem{thm}{Theorem}
\newtheorem{lem}{Lemma}

\def\llto{\mathbin{-\mkern-3mu\circ}}              % \llto is  '-o'

\def\lltensor{\otimes}

\newcommand{\Let}[3]{\mbox{\sf let\ } #1 \mbox{\ \sf be\ } #2 
                 \mbox{\ \sf in\ } #3}

\newcommand{\KB}{$\mathcal{K}_{\mathcal{B}}$}
\newcommand{\TM}{$\mathcal{T}(M)$}
\newcommand{\AAA}{$\mathcal{A}$}

\newcommand {\fli} {\Rightarrow}
\newcommand{\ldlalb}{\ma{L}_{DLAL_{B}}}   %% language of DLALB formulas
\def\ci#1{\underline{#1}}                 %% notation for Church Integer
\def\FV#1{FV({#1})}           % variables libres de t
\def\no#1#2{no({#1, #2})}     % nombre d'occurrences d'une variable dans un terme
\newcommand{\ma}{\mathcal}

\def\sub#1#2 {[ #1 \slash #2 ]}                                            %% substitution

\newcommand {\fm} {\multimap}
\newcommand{\te} {\otimes}
\newcommand{\pa} {\mathord{\S}}

\newcommand{\Ga} {\Gamma}
\newcommand{\De} {\Delta}
\newcommand{\al} {\alpha}
\newcommand{\la} {\lambda}
\newcommand{\de} {\delta}
\newcommand{\be} {\beta}

\newtheorem{definition1}{Definition}

\newcommand{\DLAL}{\mbox{\bf $DLAL$}}
\newcommand{\DLALB}{\mbox{\bf $DLAL_B$}}
\newcommand{\STA}{\mbox{\bf $STA$}}
\newcommand{\STAB}{\mbox{\bf $STA_B$}}
\newcommand{\LLLL}{\mbox{\bf $LLL$}}

\newcommand{\BOO}{\mathbf{Bool}}
\newcommand{\WWW}{\mathbf{W}}

\newcommand{\NNN}{\mathbf{N}}

\newcommand{\LaB}{\Lambda_{B}}

\begin{document}

\title{A type system for PSPACE derived from light linear logic}

\author{Lucien Capdevielle
\thanks{Partially supported by the project ANR-08-BLANC-0211-01 "COMPLICE".}
 \institute{ENS de Lyon}
 }
\def\titlerunning{A type system for PSPACE}
\def\authorrunning{Lucien Capdevielle}
\providecommand{\event}{DICE 2011}
\maketitle

%\ead{lucien.capdevielle@ens-lyon.org} 

\begin{abstract}
 We present a polymorphic type system for lambda calculus ensuring that well-typed
programs can be executed in polynomial space: \textit{dual light affine logic with booleans} (\DLALB).
 To build \DLALB\ we start from \DLAL\ (which has a simple type language with a linear and an intuitionistic type arrow,
as well as one modality) which characterizes FPTIME functions. In order to extend its expressiveness we add two boolean constants and a conditional constructor in the same way as with the system \STAB\ in \cite{Gabo}.

We show that the value of a well-typed term can be computed by an alternating machine in polynomial time, thus such a term
represents a program of PSPACE (given that PSPACE = APTIME (\cite{Chandra})).

 We also prove that all polynomial space decision functions can be represented in \DLALB.
 
 Therefore \DLALB\ characterizes PSPACE predicates.
 \end{abstract}

\section{Introduction}

The topic of this paper is Implicit Computational Complexity which is the field of study of calculi and languages with intrisic complexity-theoretical properties. One of the main issues of this field is to design programing languages with bounded computational complexity. Historically, there have been various approaches:
\begin{itemize}
\item restriction of recursive schemes (\cite{Bellancoo}, \cite{Leivantmar})
\item interpretation methods for first order interpretational languages (\cite{Marionmoy}, \cite{Hofma})
\item variations of linear logic and proofs-as-programs Curry-Howard correspondence (\cite{Girard95a}, \cite{AspertiRoversi02} and \cite{Lafont02})
\end{itemize}

 The latest approach has led to the design of type systems for $\la$-calculus such that the set of all well-typed terms corresponds to the class PTIME.
 
 In this paper, we will present a type assignment system which guarantees that a program of the language is PSPACE and that all predicates of PSPACE can be encoded in this language.

 Coming back to the approach of linear logic, it is based on the observation that the duplication rule is controlled by the logical connective "!". Moreover, the power of duplication is responsible for the complexity of normalization. Thus, by replacing the "!" with a weaker connective, one obtains systems with controlled duplications, and where normalization offers a complexity bound. Light Linear Logic (LLL, \cite{Girard95a}) and Soft Linear Logic (SLL, \cite{Lafont02}) are two examples of such systems.

	First, the system \DLAL\ (\cite{Bai} and \cite{BaillotTerui04}) has been derived from LLL and then the system \STA\ (\cite{SimonaGabo07}) from SLL. These systems are both characterizing PTIME. Then, in order to characterize PSPACE predicates, Gaboardi and al. have designed the system \STAB\ (\cite{Gabo}) by adding two boolean constants and a conditional constructor to the system \STA. The goal of this paper is to see if it is possible to adapt this method in order to obtain a system characterizing PSPACE by modifying the system \DLAL.

	It is straightforward to define \DLALB\ starting from \DLAL\ in an analogous way of \STAB\ is defined from \STA. However, proving that the complexity bound of this system is polynomial is not obvious. In fact, one difficulty is that the complexity bound of LLL and \DLAL\ is proved by using a specific reduction strategy (level-by-level strategy) which is not compatible with the conditional we add to the language. Thus we will introduce an abstract alternating machine and a measure on the terms in order to prove the PSPACE bound. Thus we use the fact that PSPACE = APTIME (\cite{Chandra}) both in the completeness and the soundness parts of the proof (contrary to the proof that \STAB\ characterizes the predicates of PSPACE where PSPACE = APTIME is only used for the completeness).

	The paper is organized as follows. We first give the definition of the system \DLALB\ and some properties of this system in Section 2. Then in Section 3 we give the proof that any well-typed term represents a predicate of APTIME. Finally, in Section 4 we prove that any predicate of APTIME is represented by a well-typed term.

%------------------------------------------------------------------------- 

%%%%%%%%%%%%%%%%%%%%%%%%%%%%%%%%%%%%%%%%%%%%%%%%%%%%%%%%%%%%%%%%%%%%%%%%%%%%%%%%%%%%%%%%%%%%%%%%%%%%%%%%%%%%%%%%%%%%%%%%%%%%%%%%%%%%%%%%%%%%%%%%%%%%%%%%%%
%%%%%%%%%%%%%%%%%%%%%%%%%%%%%%%%%%%%%%%%%%%%%%%%%%%%%%%%%%%%%%%%%%%%%%%%%%%%%%%%%%%%%%%%%%%%%%%%%%%%%%%%%%%%%%%%%%%%%%%%%%%%%%%%%%%%%%%%%%%%%%%%%%%%%%%%%%
%%%%%%%%%%%%%%%%%%%%%%%%%%%%%%%%%%%%%%%%%%%%%%%%%%%%%%%%%%%%%%%%%%%%%%%%%%%%%%%%%%%%%%%%%%%%%%%%%%%%%%%%%%%%%%%%%%%%%%%%%%%%%%%%%%%%%%%%%%%%%%%%%%%%%%%%%%
%%%%%%%%%%%%%%%%%%%%%%%%%%%%%%%%%%%%%%%%%%%%%%%%%%%%%%%%%%%%%%%%%%%%%%%%%%%%%%%%%%%%%%%%%%%%%%%%%%%%%%%%%%%%%%%%%%%%%%%%%%%%%%%%%%%%%%%%%%%%%%%%%%%%%%%%%%
%%%%%%%%%%%%%%%%%%%%%%%%%%%%%%%%%%%%%%%%%%%%%%%%%%%%%%%%%%%%%%%%%%%%%%%%%%%%%%%%%%%%%%%%%%%%%%%%%%%%%%%%%%%%%%%%%%%%%%%%%%%%%%%%%%%%%%%%%%%%%%%%%%%%%%%%%%
%%%%%%%%%%%%%%%%%%%%%%%%%%%%%%%%%%%%%%%%%%%%%%%%%%%%%%%%%%%%%%%%%%%%%%%%%%%%%%%%%%%%%%%%%%%%%%%%%%%%%%%%%%%%%%%%%%%%%%%%%%%%%%%%%%%%%%%%%%%%%%%%%%%%%%%%%%
%%%%%%%%%%%%%%%%%%%%%%%%%%%%%%%%%%%%%%%%%%%%%%%%%%%%%%%%%%%%%%%%%%%%%%%%%%%%%%%%%%%%%%%%%%%%%%%%%%%%%%%%%%%%%%%%%%%%%%%%%%%%%%%%%%%%%%%%%%%%%%%%%%%%%%%%%%
%%%%%%%%%%%%%%%%%%%%%%%%%%%%%%%%%%%%%%%%%%%%%%%%%%%%%%%%%%%%%%%%%%%%%%%%%%%%%%%%%%%%%%%%%%%%%%%%%%%%%%%%%%%%%%%%%%%%%%%%%%%%%%%%%%%%%%%%%%%%%%%%%%%%%%%%%%
%%%%%%%%%%%%%%%%%%%%%%%%%%%%%%%%%%%%%%%%%%%%%%%%%%%%%%%%%%%%%%%%%%%%%%%%%%%%%%%%%%%%%%%%%%%%%%%%%%%%%%%%%%%%%%%%%%%%%%%%%%%%%%%%%%%%%%%%%%%%%%%%%%%%%%%%%%

\section{$\la$-calculus with booleans and type assignment}

In this section, we will first define \DLALB, then we will give some classical properties which are true for terms well-typed in \DLALB.

\subsection{Definition of $\LaB$ and \DLALB}

 We start from the $\la$-calculus of \DLAL\ and will extend it with booleans and a conditional constructor in order to obtain \DLALB\ (analogous to \cite{Gabo}).\\

 The language $\ldlalb$ of \DLALB\ types is given by:
$$A, B::= \alpha \; | \; A \fm B \; | \; A \fli B \; |\; \pa A   \;|\; \forall \al . A \;|\; Bool.$$

\DLALB\ can be seen as a refinement of System F ensuring some complexity properties.

 The language $\LaB$ of $\la$-terms with booleans is given by:
$$t, u, v::= x \; | \; F \; | \; T \; | \; \la x. t \; | \; t\ u \; |\; if\ t\ then\ u\ else\ v.$$

 The terms of $\LaB$ admit another type of reduction than the $\beta$-reduction, the $\de$-reduction which is the contextual closure of:
\begin{center}
$(if\ T\ then\ u\ else\ v) \xrightarrow{\de} u$\\
and\\
$(if\ F\ then\ u\ else\ v) \xrightarrow{\de} v.$
\end{center}

\begin{definition1}\label{canoncomp}
A term $t$ of $\LaB$ can be written in a unique way as $M=N_0\ N_1\ ...\ N_m$ with $m \in \mathbb{N}$ and ($N_0=x$ or $N_0=\la x. t$ or $N_0=if\ M_0\ then\ M_1\ else\ M_2$).\\
The terms $N_i$ are called elements of the canonical composition.
\end{definition1}

In order to prove the complexity bound, we have to adapt the classical notion of number of occurences in such a way that it is compatible with the additive rule ($B$ e) of \DLALB\ (defined in Figure \ref{DLALBrules}).

\begin{definition1}\label{numboccterm}
The number of occurences of a variable in a term is inductively defined on the structure of the terms as follows: $\no{x}{x}=1$, $\no{x}{y}=0$, $\no{x}{F}=0$, $\no{x}{T}=0$, $\no{x}{\la y. t}=\no{x}{t}$, $\no{x}{\la x. t}=\no{x}{t}$, $\no{x}{t\ u}=\no{x}{t}+\no{x}{u}$, $\no{x}{if\ t_0\ then\ t_1\ else\ t_2}=\displaystyle \max_{i} \no{x}{t_i}$.
\end{definition1}

\textbf{Examples}: $\no{x}{(if\ x\ then\ x\ else\ x\ y)\ y}=1$\\
          $\no{y}{(if\ x\ then\ x\ else\ x\ y)\ y}=2$.

For \DLALB\ typing we will handle judgements of the form $\Gamma ; \Delta \vdash t:A$ (and $\Gamma \vdash_{F} t:A$ for System F).
The intended meaning is that variables in $\Delta$ are (affine)
linear, that is to say that they have at most one occurrence in the
term, while variables in $\Gamma$ are non-linear.
 We give the typing rules as a natural deduction system:
see Figure \ref{DLALBrules} (the rules of \DLALB\ are those of \DLAL\ plus ($B_0$ i), ($B_1$ i) and ($B$ e)). 
 
 We have:
\begin{itemize}
\item for ($\forall$ i): \quad (*) $\alpha$ does not
  appear free in $\Gamma, \Delta$.
\item in the ($\fli$ e) rule the r.h.s.\ premise can also be of the
  form $; \vdash u:A$ ($u$ has no free variable).
\end{itemize}

 \begin{figure*}

  \begin{center}
\fbox{
\begin{tabular}{c@{}cc}
  
  & \multicolumn{2}{c}{ {\infer[\mbox{(Id)}]{; x:A \vdash x:A}{}} }\\
  &&\\

 &{\infer[\mbox{($\fm$ i)}]{\Gamma; \Delta \vdash \la x. t: A \fm B }
 {\Gamma; \Delta, x:A \vdash t:B}}
  & {\infer[\mbox{($\fm$ e)}]{\Gamma_1,\Gamma_2; \Delta_1, \Delta_2 \vdash t\ u :B }
  {\Gamma_1; \Delta_1 \vdash t:A \fm B & \Gamma_2; \Delta_2 \vdash u:A}}\\[1ex]

&{\infer[\mbox{($\fli$ i)}]{\Gamma; \Delta \vdash \la x. t: A \fli B }
 {\Gamma, x:A ; \Delta\vdash t:B}}
  & {\infer[\mbox{($\fli$ e)}]{\Gamma, z:C ; \Delta \vdash t\ u :B }
  {\Gamma; \Delta \vdash t:A \fli B & ; z:C \vdash u:A}}\\[1ex]

&{\infer[\mbox{(Weak)}]{\Gamma_1, \Gamma_2; \Delta_1, \Delta_2 \vdash t: A }
 {\Gamma_1; \Delta_1 \vdash t:A}}
  &{\infer[\mbox{(Cntr)}]{x:A, \Gamma; \Delta \vdash t[x \slash x_1, x \slash x_2] :B }{x_1:A,x_2:A, \Gamma; \Delta \vdash t:B }} \\[1ex]

&{\infer[\mbox{($\pa$ i)}]{\Gamma;\pa\Delta\vdash t: \pa A }
 { ;\Gamma, \Delta \vdash t:A}}
  & {\infer[\mbox{($\pa$ e)}]{\Gamma_1,\Gamma_2; \Delta_1, \Delta_2 \vdash t[u \slash x] :B }
  {\Gamma_1; \Delta_1 \vdash u: \pa A  & \Gamma_2; x:\pa A, \Delta_2 \vdash t:B}}\\[1ex]

&{\infer[\mbox{($\forall$ i) (*)}]{\Gamma; \Delta \vdash  t:\forall \alpha. A}{\Gamma; \Delta \vdash t:A}} 
  & {\infer[\mbox{($\forall$ e)}]{\Gamma; \Delta \vdash t:A[B \slash \al] }
    {\Gamma; \Delta \vdash t:\forall \al. A}}\\
    &&\\

 & {\infer[\mbox{($B_0$ i)}]{ ; \vdash F:Bool }{}}
  &{\infer[\mbox{($B_1$ i)}]{ ; \vdash T:Bool }{}}\\[1ex]

 & \multicolumn{2}{c}{ {\infer[\mbox{($B$ e)}]{\Gamma; \Delta \vdash if\ M_0\ then\ M_1\ else\ M_2 :A }
    {\Gamma; \Delta \vdash M_0:\pa^n Bool & \Gamma; \Delta \vdash M_1:A & \Gamma; \Delta \vdash M_2:A & n\in \mathbb{N}}} }\\

 \end{tabular}
}
\end{center}
  \caption{Natural deduction system for \DLALB}\label{DLALBrules}
\end{figure*}

\begin{definition1}\label{depderivation}
The depth of a $\DLALB$ derivation $\mathcal{D}$ is the maximal number of premises of $(\pa i)$ and r.h.s. premises of $(\fli e)$ in a branch of $\mathcal{D}$.
\end{definition1}

\begin{definition1}\label{parderivation}
The l.h.s. premises of $(\fm e)$, $(\fli e)$ and $(\pa e)$ as well as the unique premise of $(\forall e)$ are called major premises. A \DLALB derivation is $\forall \pa$-normal if:
\begin{itemize}
\item no conclusion of a $(\forall i)$ rule is the premise of a $(\forall e)$ rule;
\item no conclusion of a $(\pa i)$ rule is the major premise of a $(\pa e)$ rule;
\item no conclusion of $(Weak)$, $(Cntr)$ and $(\pa e)$ is the major premise of elimination rules: $(\fm e)$, $(\fli e)$, $(\pa e)$ and $(\forall e)$.
\end{itemize}
\end{definition1}

\begin{definition1}\label{deltabarreduc}

%\begin{center}
Let $\bar{\de}$-reduction be the reduction defined by:\\
Let $t_0$ be a closed term.\\
Let $C$ be a context.\\
$C[if\ t_0\ then\ t_1\ else\ t_2] \xrightarrow{\bar{\de}} t_0$\\
$C[if\ t_0\ then\ t_1\ else\ t_2] \xrightarrow{\bar{\de}} C[t_1]$\\
$C[if\ t_0\ then\ t_1\ else\ t_2] \xrightarrow{\bar{\de}} C[t_2]$
%\end{center}

\end{definition1}

\textbf{Examples}:\\
$(\la x.(if\ (\la z.z\ F)\ then\ (x\ u)\ else\ y)\ v) \xrightarrow{\bar{\de}} (\la z.z\ F)$\\
$(\la x.(if\ (\la z.z\ F)\ then\ (x\ u)\ else\ y)\ v) \xrightarrow{\bar{\de}} ((\la x.(x\ u))\ v)$\\
$(\la x.(if\ (\la z.z\ F)\ then\ (x\ u)\ else\ y)\ v) \xrightarrow{\bar{\de}} ((\la x.y)\ v)$.

%%%%%%%%%%%%%%%%%%%%%%%%%%%%%%%%%%%%%%%%%%%%%%%%%%%%%%%%%%%%%%%%%%%%%%%%%%%%%%%%%%%%%%%%%%%%%%%%%%%%%%%%%%%%%%%%%%%%%%%%%%%%%%%%%%%%%%%%%%%%%%%%%%%%%%%%%%
%%%%%%%%%%%%%%%%%%%%%%%%%%%%%%%%%%%%%%%%%%%%%%%%%%%%%%%%%%%%%%%%%%%%%%%%%%%%%%%%%%%%%%%%%%%%%%%%%%%%%%%%%%%%%%%%%%%%%%%%%%%%%%%%%%%%%%%%%%%%%%%%%%%%%%%%%%
%%%%%%%%%%%%%%%%%%%%%%%%%%%%%%%%%%%%%%%%%%%%%%%%%%%%%%%%%%%%%%%%%%%%%%%%%%%%%%%%%%%%%%%%%%%%%%%%%%%%%%%%%%%%%%%%%%%%%%%%%%%%%%%%%%%%%%%%%%%%%%%%%%%%%%%%%%
%%%%%%%%%%%%%%%%%%%%%%%%%%%%%%%%%%%%%%%%%%%%%%%%%%%%%%%%%%%%%%%%%%%%%%%%%%%%%%%%%%%%%%%%%%%%%%%%%%%%%%%%%%%%%%%%%%%%%%%%%%%%%%%%%%%%%%%%%%%%%%%%%%%%%%%%%%
%%%%%%%%%%%%%%%%%%%%%%%%%%%%%%%%%%%%%%%%%%%%%%%%%%%%%%%%%%%%%%%%%%%%%%%%%%%%%%%%%%%%%%%%%%%%%%%%%%%%%%%%%%%%%%%%%%%%%%%%%%%%%%%%%%%%%%%%%%%%%%%%%%%%%%%%%%
%%%%%%%%%%%%%%%%%%%%%%%%%%%%%%%%%%%%%%%%%%%%%%%%%%%%%%%%%%%%%%%%%%%%%%%%%%%%%%%%%%%%%%%%%%%%%%%%%%%%%%%%%%%%%%%%%%%%%%%%%%%%%%%%%%%%%%%%%%%%%%%%%%%%%%%%%%
%%%%%%%%%%%%%%%%%%%%%%%%%%%%%%%%%%%%%%%%%%%%%%%%%%%%%%%%%%%%%%%%%%%%%%%%%%%%%%%%%%%%%%%%%%%%%%%%%%%%%%%%%%%%%%%%%%%%%%%%%%%%%%%%%%%%%%%%%%%%%%%%%%%%%%%%%%
%%%%%%%%%%%%%%%%%%%%%%%%%%%%%%%%%%%%%%%%%%%%%%%%%%%%%%%%%%%%%%%%%%%%%%%%%%%%%%%%%%%%%%%%%%%%%%%%%%%%%%%%%%%%%%%%%%%%%%%%%%%%%%%%%%%%%%%%%%%%%%%%%%%%%%%%%%
%%%%%%%%%%%%%%%%%%%%%%%%%%%%%%%%%%%%%%%%%%%%%%%%%%%%%%%%%%%%%%%%%%%%%%%%%%%%%%%%%%%%%%%%%%%%%%%%%%%%%%%%%%%%%%%%%%%%%%%%%%%%%%%%%%%%%%%%%%%%%%%%%%%%%%%%%%

\subsection{Properties of $\DLALB$}

The contraction rule (Cntr) is used only on variables
on the l.h.s.\ of the semi-colon.  It is then straightforward to check
the following statements:
\begin{lem}\label{freevariablelemma}(Free Variable Lemma)
\begin{itemize}
\item If $\Ga;\De \vdash t:A$ then $\FV{t}\subset dom(\Ga)\cup dom(\De)$
\item If $\Ga;\De \vdash t:A$,  $\De'\subset \De$, $\Ga'\subset \Ga$ and $\FV{t}\subset dom(\Ga')\cup dom(\De')$ then $\Ga';\De' \vdash t:A$
\item If $\Ga;\De \vdash t:A$ and $x \in \Delta$ then we have $\no{x}{t} \leqslant 1$
\end{itemize}
\end{lem}

We can make the following remarks on \DLALB\ rules:
\begin{itemize}
\item Initially the variables are linear (rule (Id)); to convert a
  linear variable into a non-linear one we can use the ($\pa$ i)
  rule. Note that it adds a $\pa$ to the type of the result and that the
 variables that remain linear get a $\pa$ type too.
\item the ($\fm$ i) (resp. ($\fli$ i)) rule corresponds to abstraction
  on a linear variable (resp. non-linear variable);
\item observe ($\fli$ e): a term of type $A\fli B$ can only be applied
to a term $u$ with at most one occurrence of free variable.
\end{itemize}

\begin{thm}\label{subjredthm}(Subject Reduction)\\
Let $\xrightarrow{\de \be}=(\xrightarrow{\be}\cup\xrightarrow{\de})$\\
If $\Ga;\De \vdash t:A$ is derivable and $t \xrightarrow{\de \be} v$, then $\Ga;\De \vdash v:A$.\\
If $\Ga;\De \vdash t:A$ is derivable and $t \xrightarrow{\bar{\de} \be} v$, then $\Ga;\De \vdash v:A$ or $\Ga;\De \vdash v:\pa^n \BOO$.
\end{thm}
\begin{proof}\\
Almost the same as in \cite{Bai}.
\end{proof}

In order to prove the strong normalisation of the terms well-typed in \DLALB, we will prove that such terms can be translated into terms of System F (which has the property of strong normalisation).

\begin{definition1}\label{traductypsystemF}
The translation ()* of a $\DLALB$ type in a type of System F is inductively defined on the structure of the types as follows: $(\alpha)^*=\alpha$, $(A \fm B)^*=(A)^*\rightarrow (B)^*$, $(A \fli B)^*=(A)^*\rightarrow (B)^*$, $(\pa A)^*=(A)^*$, $(\forall \al . A)^*=\forall \al . (A)^*$, $(Bool)^*=\forall \al . \al\rightarrow \al\rightarrow \al$.
\end{definition1}

\begin{definition1}\label{traductermsystemF}
The translation ()* of a term of $\LaB$ in a term of $\Lambda$ is inductively defined on the structure of the terms as follows: $(x)^*=x$, $(F)^*=\la x. \la y. y$, $(T)^*=\la x. \la y. x$, $(\la x. t)^*=\la x. (t)^*$, $(t\ u)^*=(t)^*\ (u)^*$, $(if\ t\ then\ u\ else\ v)^*=(t)^*\ (u)^*\ (v)^*$.
\end{definition1}

\begin{lem}\label{typdlalbtypF}
  If $\Ga;\De \vdash t:A$ then $(\Ga)^*,(\De)^* \vdash_{F} (t)^*:(A)^*$.
\end{lem}
\begin{proof}
By induction on the structure of the type derivation of $t$.
\end{proof}

\begin{lem}\label{typdlalbtypFredux}
  Let $t$ and $t'$ be two terms of $\LaB$ such that $\Ga;\De \vdash t:A$ and $t \xrightarrow{\de \be} t'$ then: $(t)^* \xrightarrow{\be} (t')^*$, $(\Ga)^*,(\De)^* \vdash_{F} (t)^*:(A)^*$ and $(\Ga)^*,(\De)^* \vdash_{F} (t')^*:(A)^*$.
\end{lem}
\begin{proof}
By the definition of the translation of the terms, Lemma \ref{typdlalbtypF} and Theorem \ref{subjredthm}.
\end{proof}

\begin{thm}\label{strongnorm}(Strong Normalisation)\\
  Let $t$ be a term of $\LaB$, if $\Ga;\De \vdash t:A$ then $t$ is strongly normalizable.
\end{thm}
\begin{proof}
By Lemma \ref{typdlalbtypFredux} and the property of strong normalization of terms typeable in System F.
\end{proof}

\begin{thm}\label{confl}(Confluence)\\
  The $\de \be$-reduction is confluent on the terms of $\LaB$ typeable in $\DLALB$.
\end{thm}
\begin{proof}
By Theorem \ref{strongnorm} and the local confluence of the $\de \be$-reduction on $\LaB$.
\end{proof}

\begin{thm}\label{normform}(Normal Form)\\
  Let $t$ be a term of $\LaB$, if $\Ga;\De \vdash t:A$ then $t$ has a unique normal form (denoted $Norm(t)$).
\end{thm}
\begin{proof}
By Theorems \ref{strongnorm} and \ref{confl}.
\end{proof}

\begin{lem}\label{normalformlemma}
  If $; \vdash t:\pa^n\BOO$ then:
\begin{enumerate}
\item $t$ is not an abstraction
\item if $t$ is normal for the $\be\de$-reduction then $t=T$ or $t=F$.
\end{enumerate}
\end{lem}

\begin{proof}
\begin{enumerate}
\item By induction on the structure of derivations.\\
\item By induction on the structure of terms and (i).
\end{enumerate}
\end{proof}

%%%%%%%%%%%%%%%%%%%%%%%%%%%%%%%%%%%%%%%%%%%%%%%%%%%%%%%%%%%%%%%%%%%%%%%%%%%%%%%%%%%%%%%%%%%%%%%%%%%%%%%%%%%%%%%%%%%%%%%%%%%%%%%%%%%%%%%%%%%%%%%%%%%%%%%%%%
%%%%%%%%%%%%%%%%%%%%%%%%%%%%%%%%%%%%%%%%%%%%%%%%%%%%%%%%%%%%%%%%%%%%%%%%%%%%%%%%%%%%%%%%%%%%%%%%%%%%%%%%%%%%%%%%%%%%%%%%%%%%%%%%%%%%%%%%%%%%%%%%%%%%%%%%%%
%%%%%%%%%%%%%%%%%%%%%%%%%%%%%%%%%%%%%%%%%%%%%%%%%%%%%%%%%%%%%%%%%%%%%%%%%%%%%%%%%%%%%%%%%%%%%%%%%%%%%%%%%%%%%%%%%%%%%%%%%%%%%%%%%%%%%%%%%%%%%%%%%%%%%%%%%%
%%%%%%%%%%%%%%%%%%%%%%%%%%%%%%%%%%%%%%%%%%%%%%%%%%%%%%%%%%%%%%%%%%%%%%%%%%%%%%%%%%%%%%%%%%%%%%%%%%%%%%%%%%%%%%%%%%%%%%%%%%%%%%%%%%%%%%%%%%%%%%%%%%%%%%%%%%
%%%%%%%%%%%%%%%%%%%%%%%%%%%%%%%%%%%%%%%%%%%%%%%%%%%%%%%%%%%%%%%%%%%%%%%%%%%%%%%%%%%%%%%%%%%%%%%%%%%%%%%%%%%%%%%%%%%%%%%%%%%%%%%%%%%%%%%%%%%%%%%%%%%%%%%%%%
%%%%%%%%%%%%%%%%%%%%%%%%%%%%%%%%%%%%%%%%%%%%%%%%%%%%%%%%%%%%%%%%%%%%%%%%%%%%%%%%%%%%%%%%%%%%%%%%%%%%%%%%%%%%%%%%%%%%%%%%%%%%%%%%%%%%%%%%%%%%%%%%%%%%%%%%%%
%%%%%%%%%%%%%%%%%%%%%%%%%%%%%%%%%%%%%%%%%%%%%%%%%%%%%%%%%%%%%%%%%%%%%%%%%%%%%%%%%%%%%%%%%%%%%%%%%%%%%%%%%%%%%%%%%%%%%%%%%%%%%%%%%%%%%%%%%%%%%%%%%%%%%%%%%%
%%%%%%%%%%%%%%%%%%%%%%%%%%%%%%%%%%%%%%%%%%%%%%%%%%%%%%%%%%%%%%%%%%%%%%%%%%%%%%%%%%%%%%%%%%%%%%%%%%%%%%%%%%%%%%%%%%%%%%%%%%%%%%%%%%%%%%%%%%%%%%%%%%%%%%%%%%
%%%%%%%%%%%%%%%%%%%%%%%%%%%%%%%%%%%%%%%%%%%%%%%%%%%%%%%%%%%%%%%%%%%%%%%%%%%%%%%%%%%%%%%%%%%%%%%%%%%%%%%%%%%%%%%%%%%%%%%%%%%%%%%%%%%%%%%%%%%%%%%%%%%%%%%%%%

\subsection{Stratified terms}

We have to describe the size of a term in detail in order to better control it during $\beta$- and $\de$-reduction.

\begin{definition1}\label{stratterm}
A stratified term is a term with each abstraction symbol $\la$ annotated by a natural number k (called its depth) and also possibly by symbol !, and with applications possibly annotated by !.
\end{definition1}

	Thus an abstraction looks like $\la^{k} x.t$ or $\la^{k!} x.t$ and an application like $t\ u$ or $t\ !\ u$. When $t$ is a stratified term, $t[+1]$ denotes $t$ with the depths of all abstraction subterms increased by 1. The type assignment rules for stratified terms are obtained by modifying some of the rules of \DLALB\ as follows:

\begin{tabular}{c@{}cc}

 &{\infer[\mbox{($\fm$ i)}]{\Gamma; \Delta \vdash \la^0 x. t: A \fm B }
 {\Gamma; \Delta, x:A \vdash t:B}}
  &{\infer[\mbox{($\fli$ i)}]{\Gamma; \Delta \vdash \la^{0!} x. t: A \fli B }
 {\Gamma, x:A ; \Delta\vdash t:B}}\\[1ex]

& {\infer[\mbox{($\fli$ e)}]{\Gamma, z:C ; \Delta \vdash t\ !\ u[+1] :B }
  {\Gamma; \Delta \vdash t:A \fli B & ; z:C \vdash u:A}}
  &{\infer[\mbox{($\pa$ i)}]{\Gamma;\pa\Delta\vdash t[+1]: \pa A }
 { ;\Gamma, \Delta \vdash t:A}}\\[1ex]  
    
 \end{tabular}

	The depth of a term is the maximal depth of all the abstractions it contains.

\begin{lem}\label{decorlemma}
  Given a \DLALB\ derivation $\mathcal{D}$ of $\Ga;\De \vdash t:A$ of depth d,
$t$ can be decorated as a stratified term $t'$ of depth d such that $\Ga;\De \vdash t':A$.
\end{lem}
\begin{proof}
By induction on the structure of the derivation $\mathcal{D}$.
\end{proof}

We can see that $\forall \pa$-Normalisation Lemma, Abstraction Property Lemma, Paragraph Property Lemma and Subject Reduction Theorem hold for stratified terms as well (as in \cite{Bai}).

\begin{definition1}\label{numbocclambdaterm}
The number of occurences of symbols $\la$ at depth k in a stratified term is inductively defined on the structure of the terms as follows: $\no{k}{x}=0$, $\no{k}{F}=0$, $\no{k}{T}=0$, $\no{k}{\la^k x. t}=\no{k}{t}+1$, $\no{k}{\la^p x. t}=\no{k}{t}$, $\no{k}{t\ u}=\no{k}{t}+\no{k}{u}$, $\no{k}{if\ t_0\ then\ t_1\ else\ t_2}=\displaystyle \max_{i} \no{k}{t_i}$.
The definition of the number of occurences of $if$ in a term $t$, $\no{if}{t}$, is similar.
\end{definition1}

\begin{lem}\label{applicationlemma}
  Let $t$ be a stratified term such that $\Ga;\De \vdash t:A$ is derivable. If $(v\ !\ u)$ is a subterm of $t$ then:
\begin{itemize}
\item $(FV(u)=\emptyset)\\
       or\\
       (FV(u)=\{x\}\ and\ (x\in Dom(\Ga)\ or\ x\ is\ bound\ in\ t\ by\ a\ \la\ annotated\ by\ !)\ and\\ \no{x}{u}=1)$
\item if $v=\la^{k!} x. r$ then $\forall p\le k,\ \no{p}{u}=0$
\end{itemize}
\end{lem}
\begin{proof}
By induction on the structure of the derivation and Lemma \ref{freevariablelemma}.
\end{proof}

We can now define, with the notations on a stratified term, a vector of integers which characterizes the size of the term.

\begin{definition1}\label{vectstratterm}
  Let $t$ be a stratified term,\\
  we define $vect_d(t)=(\no{0}{t},\ ...,\ \no{d}{t},\ \no{if}{t})$.
\end{definition1}

\begin{definition1}\label{infvect}
  Let $a$ and $b$ be two vectors of $\mathbb{Z}^p$, we define:
  \begin{itemize}
\item $a\le b$ if and only if $\forall k\le p,\ a_k\le b_k$;
\item $a<b$ if and only if $a\le b\ and\ a\ne b$.
  \end{itemize}
\end{definition1}

\begin{lem}\label{vectstratsubstlemma}
  If $t$ and $u$ are two stratified terms such that\\
  $r=\no{x}{t},\ a=vect_d(t)\ and\ b=vect_d(u)$, then\\
  $vect_d(t[u/x])\le a+r*b=(a_0+r*b_0,\ ...,\ a_{d+1}+r*b_{d+1})$.
\end{lem}
\begin{proof}
By induction on the structure of the term $t$.
\end{proof}

%%%%%%%%%%%%%%%%%%%%%%%%%%%%%%%%%%%%%%%%%%%%%%%%%%%%%%%%%%%%%%%%%%%%%%%%%%%%%%%%%%%%%%%%%%%%%%%%%%%%%%%%%%%%%%%%%%%%%%%%%%%%%%%%%%%%%%%%%%%%%%%%%%%%%%%%%%
%%%%%%%%%%%%%%%%%%%%%%%%%%%%%%%%%%%%%%%%%%%%%%%%%%%%%%%%%%%%%%%%%%%%%%%%%%%%%%%%%%%%%%%%%%%%%%%%%%%%%%%%%%%%%%%%%%%%%%%%%%%%%%%%%%%%%%%%%%%%%%%%%%%%%%%%%%
%%%%%%%%%%%%%%%%%%%%%%%%%%%%%%%%%%%%%%%%%%%%%%%%%%%%%%%%%%%%%%%%%%%%%%%%%%%%%%%%%%%%%%%%%%%%%%%%%%%%%%%%%%%%%%%%%%%%%%%%%%%%%%%%%%%%%%%%%%%%%%%%%%%%%%%%%%
%%%%%%%%%%%%%%%%%%%%%%%%%%%%%%%%%%%%%%%%%%%%%%%%%%%%%%%%%%%%%%%%%%%%%%%%%%%%%%%%%%%%%%%%%%%%%%%%%%%%%%%%%%%%%%%%%%%%%%%%%%%%%%%%%%%%%%%%%%%%%%%%%%%%%%%%%%
%%%%%%%%%%%%%%%%%%%%%%%%%%%%%%%%%%%%%%%%%%%%%%%%%%%%%%%%%%%%%%%%%%%%%%%%%%%%%%%%%%%%%%%%%%%%%%%%%%%%%%%%%%%%%%%%%%%%%%%%%%%%%%%%%%%%%%%%%%%%%%%%%%%%%%%%%%
%%%%%%%%%%%%%%%%%%%%%%%%%%%%%%%%%%%%%%%%%%%%%%%%%%%%%%%%%%%%%%%%%%%%%%%%%%%%%%%%%%%%%%%%%%%%%%%%%%%%%%%%%%%%%%%%%%%%%%%%%%%%%%%%%%%%%%%%%%%%%%%%%%%%%%%%%%
%%%%%%%%%%%%%%%%%%%%%%%%%%%%%%%%%%%%%%%%%%%%%%%%%%%%%%%%%%%%%%%%%%%%%%%%%%%%%%%%%%%%%%%%%%%%%%%%%%%%%%%%%%%%%%%%%%%%%%%%%%%%%%%%%%%%%%%%%%%%%%%%%%%%%%%%%%
%%%%%%%%%%%%%%%%%%%%%%%%%%%%%%%%%%%%%%%%%%%%%%%%%%%%%%%%%%%%%%%%%%%%%%%%%%%%%%%%%%%%%%%%%%%%%%%%%%%%%%%%%%%%%%%%%%%%%%%%%%%%%%%%%%%%%%%%%%%%%%%%%%%%%%%%%%
%%%%%%%%%%%%%%%%%%%%%%%%%%%%%%%%%%%%%%%%%%%%%%%%%%%%%%%%%%%%%%%%%%%%%%%%%%%%%%%%%%%%%%%%%%%%%%%%%%%%%%%%%%%%%%%%%%%%%%%%%%%%%%%%%%%%%%%%%%%%%%%%%%%%%%%%%%

\section{APTIME Soundness}

Usually, a complexity bound for $\LLLL$\ and related systems like \DLAL\ is obtained from a specific reduction strategy: the level by level strategy. Such strategy consists to reduce first redexes at level 0 then redexes at level 1 and so on. However, it is not possible to apply such strategy in the $\la$-calculus with the conditional constructor without breaking the polynomial bound. This is why like Gaboardi and al. we consider a $\la$-calculus machine to reduce the terms. A delicate point however is that previous work on $\LLLL$\ and \DLAL\ does not provide complexity bounds on $\la$-calculus machines. Thus, we need to introduce a suitable measure in order to prove this complexity bound.

\subsection{Definitions}

\begin{definition1}\label{program}(Programs)\\
A program is a term $t$ of $\LaB$ such that $; \vdash t:\pa^n \BOO$.\\
We define the relation $\leftarrow$ by:
\begin{itemize}
\item If $(Norm(t)=F)$ then $(no\leftarrow t)$;
\item If $(Norm(t)=T)$ then $(yes\leftarrow t)$.
\end{itemize}
\end{definition1}

\begin{definition1}\label{context}(Contexts)
\begin{itemize}
\item A context \AAA\ is a sequence of variable assignments of the shape $x_i:=t_i$ where all variables $x_i$ are distinct. The set of contexts is denoted by $C_{tx}$.
\item The cardinality of a context \AAA, denoted by $\#($\AAA$)$, is the number of variable assignments in \AAA.
\item The empty context is denoted by $\emptyset$.
\item Let \AAA$=[x_1:=t_1,\ ...,\ x_n:=t_n]$ be a context.
Then $()^{\mathcal{A}} : \LaB \rightarrow \LaB$ is the map that associates the term $t[t_n/x_n]...[t_1/x_1]$ to each term $t$.
\end{itemize}
\end{definition1}

\begin{definition1}\label{configuration}(Configurations)\\
 There is 4 types of configurations:
 \begin{itemize}
\item a rejecting configuration: $\llbracket (Rejecting)\rrbracket$;
\item an accepting configuration: $\llbracket (Accepting)\rrbracket$;
\item an existensial configuration: $\llbracket(\exists)\ $\AAA$\ |\ \{b;\ t\}\rrbracket$ with \AAA\ a context, $t$ a term and $b\in\{yes;no\}$;
\item a universal configuration: $\llbracket(\forall)\ $\AAA$\ |\ \{b;\ t\}\ \{b';\ t'\}\rrbracket$ with \AAA\ a context, $t$ and $t'$ two terms and $b,b'\in\{yes;no\}$;
\end{itemize}
\end{definition1}

\begin{definition1}\label{AAM}
The Abstract Alternating Machine \KB\ (which is similar to the Krivine machine (\cite{Krivine}) when restricted to the $\lambda$-calculus) is a machine that takes as input a program $t$, starts with the initial configuration
$\llbracket(\exists)\ \emptyset\ |\ \{yes;\ t\}\rrbracket$ and reduces t using the two transition functions described in Figure \ref{KBrules}.
It accepts the program $t$ if its normal form is true and rejects it if its normal form is false (as will be shown below).
\end{definition1}

The base cases are obvious.
The $(\be)$ transition applies when the head of the subject is a $\be$-redex. Then the association between the bound variable and the argument is remembered in context \AAA.
The $(h)$ transition replaces the head occurence of the head variable by the term associated with it in the context.
The $(if)$ transitions, always followed by the $(if')$ transitions, perform the $\bar{\de}$ reductions (following the intuition that: $if\ t_0\ then\ t_1\ else\ t_2 = (t_0 \land t_1) \lor (\lnot t_0 \land t_2)$).

\begin{definition1}\label{computation}(Computations)\\
The computation of the Abstract alternating machine \KB\ is the tree obtained by applying the rules given in figure \ref{KBrules} starting from the initial configuration. The definition of a configuration accepted by \KB\ and of a computation accepted by \KB\ is the same as those of the Alternating Turing Machine.
\end{definition1}

 \begin{figure*}

  \begin{center}
\fbox{
\begin{tabular}{c}

\begin{tabular}{rcl}

  \scriptsize{$\llbracket(\exists)\ $\AAA$\ |\ \{b;\ \la x.N\ N_1\ ...\ N_p\}\rrbracket$} & $\xrightarrow[(\beta)]{1/2}$ &
  \scriptsize{$\llbracket(\exists)\ $\AAA$@(x':=N_1)\ |\ \{b;\ N[x'/x]\ N_2\ ...\ N_p\}\rrbracket(*)$}\\

  \scriptsize{$\llbracket(\exists)\ $\AAA$_1@(x:=N)@$\AAA$_2\ |\ \{b;\ x\ N_1\ ...\ N_p\}\rrbracket$} & $\xrightarrow[(h)]{1/2}$ &
  \scriptsize{$\llbracket(\exists)\ $\AAA$_1@(x:=N)@$\AAA$_2\ |\ \{b;\ N\ N_1\ ...\ N_p\}\rrbracket$}\\

  \scriptsize{$\llbracket(\exists)\ $\AAA$\ |\ \{b;\ (if\ M_0\ then\ M_1\ else\ M_2)\ N_1\ ...\ N_p\}\rrbracket$} & $\xrightarrow[(if)]{1}$ &
  \scriptsize{$\llbracket(\forall)\ $\AAA$\ |\ \{yes;\ M_0\}\ \{b;\ M_1\ N_1\ ...\ N_p\}\rrbracket$}\\

  \scriptsize{$\llbracket(\exists)\ $\AAA$\ |\ \{b;\ (if\ M_0\ then\ M_1\ else\ M_2)\ N_1\ ...\ N_p\}\rrbracket$} & $\xrightarrow[(if)]{2}$ &
  \scriptsize{$\llbracket(\forall)\ $\AAA$\ |\ \{no;\ M_0\}\ \{b;\ M_2\ N_1\ ...\ N_p\}\rrbracket$}\\

  \scriptsize{$\llbracket(\forall)\ $\AAA$\ |\ \{a;\ M_0\}\ \{b;\ N\}\rrbracket$} & $\xrightarrow[(if')]{1}$ &
  \scriptsize{$\llbracket(\exists)\ $\AAA$\ |\ \{a;\ M_0\}\rrbracket$}\\

  \scriptsize{$\llbracket(\forall)\ $\AAA$\ |\ \{a;\ M_0\}\ \{b;\ N\}\rrbracket$} & $\xrightarrow[(if')]{2}$ &
  \scriptsize{$\llbracket(\exists)\ $\AAA$\ |\ \{b;\ N\}\rrbracket$}\\
 \end{tabular}\\

\begin{tabular}{ccc}

  base & \scriptsize{$\llbracket(\exists)\ $\AAA$\ |\ \{yes;\ T\}\rrbracket \xrightarrow{1/2} \llbracket (Accepting)\rrbracket$} &
  \scriptsize{$\llbracket(\exists)\ $\AAA$\ |\ \{no;\ F\}\rrbracket \xrightarrow{1/2} \llbracket (Accepting)\rrbracket$}\\

  cases & \scriptsize{$\llbracket(\exists)\ $\AAA$\ |\ \{no;\ T\}\rrbracket \xrightarrow{1/2} \llbracket (Rejecting)\rrbracket$} &
  \scriptsize{$\llbracket(\exists)\ $\AAA$\ |\ \{yes;\ F\}\rrbracket \xrightarrow{1/2} \llbracket (Rejecting)\rrbracket$}\\

\end{tabular}\\

\scriptsize{(*) x' is a fresh variable. $1/2$ means $1$ or $2$.}\\

\end{tabular}
}
\end{center}
  \caption{The Rules of the Abstract Alternating Machine \KB}\label{KBrules}
\end{figure*}

From here until the end of the subsection \ref{correctKBsec}, we will fix a program $M$. Note that:
\begin{itemize}
\item $m=|M|$ (with $|M|$ the size of $M$);
\item $\mathcal{D}$ is a derivation of $; \vdash M:\pa^n \BOO$;
\item $d$ is the depth of $\mathcal{D}$;
\item $M'$ is the stratified term of depth $d$ associated with the term $M$;
\item $r=\displaystyle \max_{x} \no{x}{M}$ (with $r<m$ by definition).
\end{itemize}

\begin{definition1}\label{tkukvk}
Let $t_k,\ u_k,\ v_k : \mathbb{Z}^{d+2} \rightarrow \mathbb{Z}^{d+2}$ such that:
\begin{itemize}
\item $t_k(a)=(a_0,\ ...,\ a_{k-1},\ a_k-1,\ a_{k+1}+r(b_{k+1}+m),\ ...,\ a_{d+1}+r(b_{d+1}+m))$;
\item $u_k(a)=(a_0,\ ...,\ a_{k-1},\ a_k-1,\ a_{k+1},\ ...,\ a_{d+1})$;
\item $v_k(b)=(b_0,\ ...,\ b_k,\ b_{k+1}+m,\ ...,\ b_{d+1}+m)$.
\end{itemize}
\end{definition1}

We want to establish a complexity bound on the machine. For that, we define a measure on the vectors characterizing the size of terms such that this measure will decrease with $\beta$- and $\de$-reduction.

\begin{definition1}\label{measure}
Let $measure_{(i)} : \mathbb{Z}^{i+2} \times \mathbb{Z}^{i+2} \rightarrow \mathbb{Z}$ such that:
\begin{itemize}
\item $measure_{(-1)}(a_0,\ b_0)=a_0$
\item $measure_{(i+1)}((a_0,\ ...,\ a_{i+2}),\ (b_0,\ ...,\ b_{i+2}))=$\\
$measure_{(i)}((a_1+(r+1)(b_1+a_0*m)a_0,\ ...,\ a_{i+2}+(r+1)(b_{i+2}+a_0*m)a_0),$\\
            $(b_1+a_0*m,\ ...,\ b_{i+2}+a_0*m))$.
\end{itemize}
\end{definition1}

\begin{lem}\label{lemeasure}
$\forall k\ge -1,\ \forall a,b,a',b'\in \mathbb{N}^{k+2},$\\
                          $if\ a'\le a\ and\ b'\le b\ then\ measure_{(k)}(a',\ b')\le measure_{(k)}(a,\ b)$;\\
                          $if\ a'\le a,\ b'\le b\ and\ a'\ne a\ then\ measure_{(k)}(a',\ b') < measure_{(k)}(a,\ b)$.
\end{lem}
\begin{proof}
By definition of $measure$.
\end{proof}

\begin{lem}\label{tkmeasure}
Let $a, b\in \mathbb{N}^{d+2}$.
\begin{itemize}
\item $\forall k\in \llbracket 0;\ d \rrbracket ,\ if\ t_k(a),v_k(b)\in \mathbb{N}^{d+2}\ then\ 0\le measure_{(d)}(t_k(a),\ v_k(b)) < measure_{(d)}(a,\ b)$;
\item $\forall k\in \llbracket 0;\ d+1 \rrbracket ,\ if\ u_k(a)\in \mathbb{N}^{d+2}\ then\ 0\le measure_{(d)}(u_k(a),\ b) < measure_{(d)}(a,\ b)$.
\end{itemize}
\end{lem}
\begin{proof}
By Lemma \ref{lemeasure} and definitions of $measure,\ t,\ u\ and\ v$.
\end{proof}

\begin{definition1}\label{TTM}
 The Transformation Tree of $M$, \TM, describes the computation tree of \KB\ on the input $M$ and contains nodes which are elements of the set $\mathbb{N} \times (Ctx \times \LaB) \times (\mathbb{N}^{d+2} \times \mathbb{N}^{d+2})$. This tree is inductively defined by the rules given in Figure \ref{Treerules}. Note that the terms in this tree are decorated (these terms are stratified terms).
\end{definition1}

\TM\ will be used to bound the time of computation of \KB\ on $M$.

 \begin{figure*}

  \begin{center}
\fbox{
\begin{tabular}{c}

\infer[\mbox{($\be$!)}]{$\scriptsize{$\llbracket(i)\ ($\AAA$\ |\ (\la^{k!} x.N\ !\ N_1)\ ...\ N_p)\ (a,\ b)\rrbracket$}$}
                       {$\scriptsize{$\llbracket(i+1)\ ($\AAA$@(x_i^! := N_1)\ |\ N[x_i^!/x]\ N_2\ ...\ N_p)\ (t_k(a),\ v_k(b))\rrbracket$}$}\\
\\

\infer[\mbox{($\be$)}]{$\scriptsize{$\llbracket(i)\ ($\AAA$\ |\ \la^{k} x.N\ N_1\ ...\ N_p)\ (a,\ b)\rrbracket$}$}
                       {$\scriptsize{$\llbracket(i+1)\ ($\AAA$@(x_i := N_1)\ |\ N[x_i/x]\ N_2\ ...\ N_p)\ (u_k(a),\ b)\rrbracket$}$}\\
\\

\infer[\mbox{(h)}]{$\scriptsize{$\llbracket(i)\ ($\AAA$_1@(x:=N)@$\AAA$_2\ |\ x\ N_1\ ...\ N_p)\ (a,\ b)\rrbracket$}$}
                       {$\scriptsize{$\llbracket(i)\ ($\AAA$_1@(x:=N)@$\AAA$_2\ |\ N\ N_1\ ...\ N_p)\ (a,\ b)\rrbracket$}$}\\
\\

\infer[\mbox{(if)}]{$\scriptsize{$\llbracket(i)\ ($\AAA$\ |\ (if\ M_0\ then\ M_1\ else\ M_2)\ N_1\ ...\ N_p)\ (a,\ b)\rrbracket$}$}
                       {$\scriptsize{$\llbracket(i)\ ($\AAA$\ |\ M_0)\ (u_{d+1}(a),\ b)\rrbracket$}$ & $\scriptsize{$\llbracket(i)\ ($\AAA$\ |\ M_1\ N_1\ ...\ N_p)\ (u_{d+1}(a),\ b)\rrbracket$}$ & $\scriptsize{$\llbracket(i)\ ($\AAA$\ |\ M_2\ N_1\ ...\ N_p)\ (u_{d+1}(a),\ b)\rrbracket$}$}\\
\\

\infer[\mbox{(root)}]{}{$\scriptsize{$\llbracket(0)\ (\emptyset\ |\ M')\ ((m,\ ...,\ m),\ (0,\ ...,\ 0))\rrbracket$}$}\\

\end{tabular}
}
\end{center}
  \caption{The Rules of the Transformation Tree of $M$: \TM}\label{Treerules}
\end{figure*}

%%%%%%%%%%%%%%%%%%%%%%%%%%%%%%%%%%%%%%%%%%%%%%%%%%%%%%%%%%%%%%%%%%%%%%%%%%%%%%%%%%%%%%%%%%%%%%%%%%%%%%%%%%%%%%%%%%%%%%%%%%%%%%%%%%%%%%%%%%%%%%%%%%%%%%%%%%
%%%%%%%%%%%%%%%%%%%%%%%%%%%%%%%%%%%%%%%%%%%%%%%%%%%%%%%%%%%%%%%%%%%%%%%%%%%%%%%%%%%%%%%%%%%%%%%%%%%%%%%%%%%%%%%%%%%%%%%%%%%%%%%%%%%%%%%%%%%%%%%%%%%%%%%%%%
%%%%%%%%%%%%%%%%%%%%%%%%%%%%%%%%%%%%%%%%%%%%%%%%%%%%%%%%%%%%%%%%%%%%%%%%%%%%%%%%%%%%%%%%%%%%%%%%%%%%%%%%%%%%%%%%%%%%%%%%%%%%%%%%%%%%%%%%%%%%%%%%%%%%%%%%%%
%%%%%%%%%%%%%%%%%%%%%%%%%%%%%%%%%%%%%%%%%%%%%%%%%%%%%%%%%%%%%%%%%%%%%%%%%%%%%%%%%%%%%%%%%%%%%%%%%%%%%%%%%%%%%%%%%%%%%%%%%%%%%%%%%%%%%%%%%%%%%%%%%%%%%%%%%%
%%%%%%%%%%%%%%%%%%%%%%%%%%%%%%%%%%%%%%%%%%%%%%%%%%%%%%%%%%%%%%%%%%%%%%%%%%%%%%%%%%%%%%%%%%%%%%%%%%%%%%%%%%%%%%%%%%%%%%%%%%%%%%%%%%%%%%%%%%%%%%%%%%%%%%%%%%
%%%%%%%%%%%%%%%%%%%%%%%%%%%%%%%%%%%%%%%%%%%%%%%%%%%%%%%%%%%%%%%%%%%%%%%%%%%%%%%%%%%%%%%%%%%%%%%%%%%%%%%%%%%%%%%%%%%%%%%%%%%%%%%%%%%%%%%%%%%%%%%%%%%%%%%%%%
%%%%%%%%%%%%%%%%%%%%%%%%%%%%%%%%%%%%%%%%%%%%%%%%%%%%%%%%%%%%%%%%%%%%%%%%%%%%%%%%%%%%%%%%%%%%%%%%%%%%%%%%%%%%%%%%%%%%%%%%%%%%%%%%%%%%%%%%%%%%%%%%%%%%%%%%%%
%%%%%%%%%%%%%%%%%%%%%%%%%%%%%%%%%%%%%%%%%%%%%%%%%%%%%%%%%%%%%%%%%%%%%%%%%%%%%%%%%%%%%%%%%%%%%%%%%%%%%%%%%%%%%%%%%%%%%%%%%%%%%%%%%%%%%%%%%%%%%%%%%%%%%%%%%%
%%%%%%%%%%%%%%%%%%%%%%%%%%%%%%%%%%%%%%%%%%%%%%%%%%%%%%%%%%%%%%%%%%%%%%%%%%%%%%%%%%%%%%%%%%%%%%%%%%%%%%%%%%%%%%%%%%%%%%%%%%%%%%%%%%%%%%%%%%%%%%%%%%%%%%%%%%

\subsection{APTIME soundness of \KB}\label{aptimeKB}

\begin{lem}\label{positmeasure}
  Let $\llbracket (j),\ ($\AAA$|t),\ (a,\ b) \rrbracket$ be a node of \TM.
  \begin{enumerate}
  \item $for\ each\ (x_{i}^{!}:=t_i) \in \mathcal{A},\ vect_d((t_i)^{\mathcal{A}})\le b$.
  \item $vect_d((t)^{\mathcal{A}})\le a$.
  \end{enumerate}
\end{lem}
\begin{proof}
\begin{enumerate}
\item By induction on the structure of the tree using Lemmas \ref{applicationlemma} and \ref{vectstratsubstlemma}.
\item By induction on the structure of the tree using (1) and Lemma \ref{vectstratsubstlemma}\\
(given that all the elements of the canonical composition of $t$ and all the terms of \AAA\ are subterms of $M$).
\end{enumerate}
\end{proof}

\begin{lem}\label{reduxmeasure}
  Let $n=\llbracket (j),\ ($\AAA$|t),\ (a,\ b) \rrbracket$ and $n'=\llbracket (j'),\ ($\AAA$'|t'),\ (a',\ b') \rrbracket$ be two nodes of \TM.
  If $n'$ is a son of $n$ linked by a rule $(\be !)$, $(\be)$ or $(if)$, then $0\le measure_{(d)}(a',\ b') < measure_{(d)}(a,\ b)$.
\end{lem}
\begin{proof}
By Lemma \ref{positmeasure} we have $a,b,a',b'\in \mathbb{N}^{d+2}$, thus by Lemma \ref{tkmeasure} we have $0\le measure_{(d)}(a',\ b') < measure_{(d)}(a,\ b)$.
\end{proof}

\begin{definition1}\label{branchtree}
Let $n$ be a node of \TM.
\begin{itemize}
\item $\#_{\be !}(n)$ denotes the number of applications of the $(\be !)$ rule in the path between the root of \TM\ and $n$.
\item $\#_{\be}(n)$ denotes the number of applications of the $(\be)$ rule in the path between the root of \TM\ and $n$.
\item $\#_{h}(n)$ denotes the number of applications of the $(h)$ rule in the path between the root of \TM\ and $n$.
\item $\#_{if}(n)$ denotes the number of applications of the $(if)$ rule in the path between the root of \TM\ and $n$.
\end{itemize}
\end{definition1}

\begin{lem}\label{numbifbetbranch}
Let $n$ be a node of \TM.
\begin{enumerate}
\item $\#_{\be !}(n)+\#_{\be}(n)+\#_{if}(n)\le measure_{(d)}((m,\ ...,\ m),\ (0,\ ...,\ 0))$
\item $\#_{h}(n)\le (measure_{(d)}((m,\ ...,\ m),\ (0,\ ...,\ 0)))^2$
\end{enumerate}
\end{lem}
\begin{proof}
\begin{enumerate}
\item By Lemma \ref{reduxmeasure}.
\item Let $n'$ and $n''$ two nodes of \TM\ such that there is a path of applications of the $(h)$ rule from $n'$ to $n''$.\\
      $\#_{h}(n'')-\#_{h}(n')\le \#($\AAA$_{n'})=\#_{\be !}(n')+\#_{\be}(n')$.\\
      Thus $\#_{h}(n)\le (\#_{\be !}(n)+\#_{\be}(n))*(\#_{\be !}(n)+\#_{\be}(n)+\#_{if}(n))$.
\end{enumerate}
\end{proof}

\begin{lem}\label{boundmeasure}
   $\forall i\in \llbracket 1;\ d+1 \rrbracket,$\\
   $measure_{(d)}((m,...,m),\ (0,...,0))\le measure_{(d-i)}((m^{3^{i+1}},...,m^{3^{i+1}}),\ (m^{3^{i}+1},...,m^{3^{i}+1}))$.
\end{lem}
\begin{proof}
By induction on $i$ using Lemma \ref{lemeasure}.
\end{proof}

\begin{thm}\label{KBbound}
   The machine \KB\ on the input $M$ is computing in a time bounded by $m^{3^{d+3}}$.
\end{thm}
\begin{proof}
  By Lemma \ref{boundmeasure}, $measure_{(d)}((m,\ ...,\ m),\ (0,\ ...,\ 0))\le m^{3^{d+2}}$.\\
  Furthermore by Lemma \ref{numbifbetbranch}, let $n$ be a node of \TM,\\
  let $k=measure_{(d)}((m,\ ...,\ m),\ (0,\ ...,\ 0))$ and let $p_n=\#_{\be !}(n)+\#_{\be}(n)+\#_{if}(n)+\#_{h}(n)$,\\
  $p_n\le k+k^2$.\\
  Thus $Time($\KB$(M))\le 2*Depth($\TM$)+1=2*(\displaystyle \max_{n} p_n)+1\le 2*(m^{3^{d+2}}+m^{2*3^{d+2}})+1$.
\end{proof}

	Note that the measure we have defined can be applied in the restricted case of \DLAL\ programs: in this case the machine is deterministic and our measure gives a new proof that \DLAL\ terms of boolean type can be evaluated in polynomial time where the degree of the polynomial depends on the depth of the term.

%%%%%%%%%%%%%%%%%%%%%%%%%%%%%%%%%%%%%%%%%%%%%%%%%%%%%%%%%%%%%%%%%%%%%%%%%%%%%%%%%%%%%%%%%%%%%%%%%%%%%%%%%%%%%%%%%%%%%%%%%%%%%%%%%%%%%%%%%%%%%%%%%%%%%%%%%%
%%%%%%%%%%%%%%%%%%%%%%%%%%%%%%%%%%%%%%%%%%%%%%%%%%%%%%%%%%%%%%%%%%%%%%%%%%%%%%%%%%%%%%%%%%%%%%%%%%%%%%%%%%%%%%%%%%%%%%%%%%%%%%%%%%%%%%%%%%%%%%%%%%%%%%%%%%
%%%%%%%%%%%%%%%%%%%%%%%%%%%%%%%%%%%%%%%%%%%%%%%%%%%%%%%%%%%%%%%%%%%%%%%%%%%%%%%%%%%%%%%%%%%%%%%%%%%%%%%%%%%%%%%%%%%%%%%%%%%%%%%%%%%%%%%%%%%%%%%%%%%%%%%%%%
%%%%%%%%%%%%%%%%%%%%%%%%%%%%%%%%%%%%%%%%%%%%%%%%%%%%%%%%%%%%%%%%%%%%%%%%%%%%%%%%%%%%%%%%%%%%%%%%%%%%%%%%%%%%%%%%%%%%%%%%%%%%%%%%%%%%%%%%%%%%%%%%%%%%%%%%%%
%%%%%%%%%%%%%%%%%%%%%%%%%%%%%%%%%%%%%%%%%%%%%%%%%%%%%%%%%%%%%%%%%%%%%%%%%%%%%%%%%%%%%%%%%%%%%%%%%%%%%%%%%%%%%%%%%%%%%%%%%%%%%%%%%%%%%%%%%%%%%%%%%%%%%%%%%%
%%%%%%%%%%%%%%%%%%%%%%%%%%%%%%%%%%%%%%%%%%%%%%%%%%%%%%%%%%%%%%%%%%%%%%%%%%%%%%%%%%%%%%%%%%%%%%%%%%%%%%%%%%%%%%%%%%%%%%%%%%%%%%%%%%%%%%%%%%%%%%%%%%%%%%%%%%
%%%%%%%%%%%%%%%%%%%%%%%%%%%%%%%%%%%%%%%%%%%%%%%%%%%%%%%%%%%%%%%%%%%%%%%%%%%%%%%%%%%%%%%%%%%%%%%%%%%%%%%%%%%%%%%%%%%%%%%%%%%%%%%%%%%%%%%%%%%%%%%%%%%%%%%%%%
%%%%%%%%%%%%%%%%%%%%%%%%%%%%%%%%%%%%%%%%%%%%%%%%%%%%%%%%%%%%%%%%%%%%%%%%%%%%%%%%%%%%%%%%%%%%%%%%%%%%%%%%%%%%%%%%%%%%%%%%%%%%%%%%%%%%%%%%%%%%%%%%%%%%%%%%%%
%%%%%%%%%%%%%%%%%%%%%%%%%%%%%%%%%%%%%%%%%%%%%%%%%%%%%%%%%%%%%%%%%%%%%%%%%%%%%%%%%%%%%%%%%%%%%%%%%%%%%%%%%%%%%%%%%%%%%%%%%%%%%%%%%%%%%%%%%%%%%%%%%%%%%%%%%%

\subsection{Correctness of \KB}\label{correctKBsec}

Now, we need to prove that the alternating $\la$-calculus machine computes the right value.

\begin{lem}\label{boolean}
If $\llbracket(\exists)\ $\AAA$\ |\ \{b;\ t\}\rrbracket$ (resp. $\llbracket(\forall)\ $\AAA$\ |\ \{b;\ t\}\ \{b';\ t'\}\rrbracket$) is a configuration of the computation of \KB\ on $M$ then $; \vdash (t)^{\mathcal{A}}:\pa^n \BOO$ (resp. $; \vdash (t)^{\mathcal{A}}:\pa^n \BOO\ and\ ; \vdash (t')^{\mathcal{A}}:\pa^{n'} \BOO$).
\end{lem}
\begin{proof}
By induction on the structure of the tree using Theorem \ref{subjredthm}.
\end{proof}

\begin{lem}\label{correct}
If $c=\llbracket(\exists)\ $\AAA$\ |\ \{b;\ t\}\rrbracket$ (resp. $\llbracket(\forall)\ $\AAA$\ |\ \{b;\ t\}\ \{b';\ t'\}\rrbracket$) is a configuration of the computation of \KB\ on $M$ then \KB\ is accepting $c$ if and only if $b\leftarrow (t)^{\mathcal{A}}$ (resp. $b\leftarrow (t)^{\mathcal{A}}$ and $b'\leftarrow (t')^{\mathcal{A}}$).
\end{lem}
\begin{proof}
By induction on the tree, starting from the leafs and using Lemmas \ref{boolean} and \ref{normalformlemma}.
\end{proof}

\begin{thm}\label{KBcorrect}
   The machine \KB\ is accepting $M$ if and only if $Norm(M)=T$.
\end{thm}
\begin{proof}
  By Lemma \ref{correct}.
\end{proof}

If $t$ is a closed term of type $\WWW \fli \pa^n \BOO$, we define $\mathcal{L}(t)$ as the set of words accepted by $t$. Finally, we obtain the desired result:

\begin{thm}[APTIME soundness of \DLALB]\label{correctness}
   Let $t$ be a term of $\LaB$ such that $; \vdash t:\WWW \fli \pa^n \BOO$ or $; \vdash t:\WWW \fm \pa^n \BOO$ has a derivation of depth $d$.\\
   Let $\mathcal{M}$ be the Alternating Turing Machine which, on the input $i$ represented by the $\la$-term $w$, simulates the machine \KB\ on the input $(t\ w)$.\\
   Then $\mathcal{M}$ decides the language represented by $t$ and $\mathcal{M}$ is computing in time $O(m^{3^{d+4}})$.\\
   Thus $\mathcal{L}(t)\in APTIME$.
\end{thm}
\begin{proof}
  By Theorems \ref{KBcorrect} and \ref{KBbound}.
\end{proof}

%%%%%%%%%%%%%%%%%%%%%%%%%%%%%%%%%%%%%%%%%%%%%%%%%%%%%%%%%%%%%%%%%%%%%%%%%%%%%%%%%%%%%%%%%%%%%%%%%%%%%%%%%%%%%%%%%%%%%%%%%%%%%%%%%%%%%%%%%%%%%%%%%%%%%%%%%%
%%%%%%%%%%%%%%%%%%%%%%%%%%%%%%%%%%%%%%%%%%%%%%%%%%%%%%%%%%%%%%%%%%%%%%%%%%%%%%%%%%%%%%%%%%%%%%%%%%%%%%%%%%%%%%%%%%%%%%%%%%%%%%%%%%%%%%%%%%%%%%%%%%%%%%%%%%
%%%%%%%%%%%%%%%%%%%%%%%%%%%%%%%%%%%%%%%%%%%%%%%%%%%%%%%%%%%%%%%%%%%%%%%%%%%%%%%%%%%%%%%%%%%%%%%%%%%%%%%%%%%%%%%%%%%%%%%%%%%%%%%%%%%%%%%%%%%%%%%%%%%%%%%%%%
%%%%%%%%%%%%%%%%%%%%%%%%%%%%%%%%%%%%%%%%%%%%%%%%%%%%%%%%%%%%%%%%%%%%%%%%%%%%%%%%%%%%%%%%%%%%%%%%%%%%%%%%%%%%%%%%%%%%%%%%%%%%%%%%%%%%%%%%%%%%%%%%%%%%%%%%%%
%%%%%%%%%%%%%%%%%%%%%%%%%%%%%%%%%%%%%%%%%%%%%%%%%%%%%%%%%%%%%%%%%%%%%%%%%%%%%%%%%%%%%%%%%%%%%%%%%%%%%%%%%%%%%%%%%%%%%%%%%%%%%%%%%%%%%%%%%%%%%%%%%%%%%%%%%%
%%%%%%%%%%%%%%%%%%%%%%%%%%%%%%%%%%%%%%%%%%%%%%%%%%%%%%%%%%%%%%%%%%%%%%%%%%%%%%%%%%%%%%%%%%%%%%%%%%%%%%%%%%%%%%%%%%%%%%%%%%%%%%%%%%%%%%%%%%%%%%%%%%%%%%%%%%
%%%%%%%%%%%%%%%%%%%%%%%%%%%%%%%%%%%%%%%%%%%%%%%%%%%%%%%%%%%%%%%%%%%%%%%%%%%%%%%%%%%%%%%%%%%%%%%%%%%%%%%%%%%%%%%%%%%%%%%%%%%%%%%%%%%%%%%%%%%%%%%%%%%%%%%%%%
%%%%%%%%%%%%%%%%%%%%%%%%%%%%%%%%%%%%%%%%%%%%%%%%%%%%%%%%%%%%%%%%%%%%%%%%%%%%%%%%%%%%%%%%%%%%%%%%%%%%%%%%%%%%%%%%%%%%%%%%%%%%%%%%%%%%%%%%%%%%%%%%%%%%%%%%%%
%%%%%%%%%%%%%%%%%%%%%%%%%%%%%%%%%%%%%%%%%%%%%%%%%%%%%%%%%%%%%%%%%%%%%%%%%%%%%%%%%%%%%%%%%%%%%%%%%%%%%%%%%%%%%%%%%%%%%%%%%%%%%%%%%%%%%%%%%%%%%%%%%%%%%%%%%%

\section{APTIME Completeness}

This section presents the second part of the proof that \DLALB\ characterizes the predicates of PSPACE and is simply using classical ideas of the literature (see \cite{BaillotTerui04} and \cite{Gabo}).

 We have the following data types for unary integers
and binary words in \DLALB:
\begin{eqnarray*}
\NNN & = & \forall \al. (\al \fm \al) \fli \pa (\al \fm \al),\\
\WWW & = & \forall \al. (\al \fm \al) \fli
(\al \fm \al) \fli \pa (\al \fm \al).
\end{eqnarray*}

The inhabitants of types $\NNN$ and $\WWW$ are the familiar Church codings
of integers and words:
\begin{eqnarray*}
\ci{n} & = & \la f. \la x. \underbrace{f  (f \dots (f}_{n} x) \dots ),\\
\ci{w} & = & \la f_0. \la f_1. \la x. f_{i_1} (f_{i_2} \dots (f_{i_n} x) \dots ),
\end{eqnarray*}
with $i\in\{0,1\}$, $n\in N$ and $w= i_1 i_2 \cdots i_n \in \{0,1\}^*$.

It can be useful in practice to use a type $A \otimes B$. It can be defined anyway, thanks to full weakening:

$$
  A \otimes B  = \forall \alpha. ((A \fm B \fm \alpha) \fm \alpha).
$$ 

 We use as syntactic sugar the following new constructions on terms with the typing rules of
 Figure \ref{derivedrules2}:
$$
\begin{array}{rcl}
 t_1 \te t_2 & = & \la x. x t_1 t_2, \\
\Let{u}{x_1 \te x_2}{t} & = & u (\la x_1. \la x_2. t).
\end{array}
$$
\begin{figure*}

  \begin{center}
\fbox{
\begin{tabular}{c}
   \infer[\mbox{($\otimes$\; i)}]{\Gamma_1, \Gamma_2 ; \Delta_1, \Delta_2 \vdash t_1 \te t_2:A_1 \te A_2 }
  {\Gamma_1; \Delta_1 \vdash t_1:A_1 & \Gamma_2; \Delta_2 \vdash t_2:A_2}
 \; \\[1em]
  \infer[\mbox{($\otimes$\; e)}]{\Gamma_1, \Gamma_2;  \Delta_1, \Delta_2 \vdash \Let{u}{x_1 \te x_2}{t}: B }
 { \Gamma_1; \Delta_1 \vdash u:A_1 \otimes A_2  & \Gamma_2; x_1:A_1, x_2:A_2, \Delta_2 \vdash t:B} 
\end{tabular}
}
\end{center}
  \caption{Derived rules}\label{derivedrules2}
\end{figure*}

\begin{thm}[APTIME completeness of \DLALB]\label{t-fpcompletelalc}
  If a function $f: \{0,1\}^{\star} \rightarrow \{0,1\}$ is
  computable in time $n^{2^d}$ by a one-tape alternating Turing machine for some $d$, 
then there exists a term $M$ of $\LaB$
  such that $;\vdash M: \WWW \fli \pa^{2d+2} \BOO$ and
  $M$ represents $f$.
\end{thm}

\noindent
{\em Proof (Sketch).} Let $\mathcal{M}$ be an ATM with $2$ symbols, $1$ tape, $k$ classical states and four characteristic states.
 The four characteristic states: Accepting, Rejecting, Universal and Existential, are represented respectively by
$A = F\ \lltensor\ T$, $R = F\ \lltensor\ F$, $\wedge\ = T\ \lltensor\ F$ and $\vee\ = T\ \lltensor\ T$ of type $\BOO^2$.
 
Following the idea of \cite{AspertiRoversi02},
let $\mathbf{Conf}$ be the \DLALB-type
$$
\forall \alpha. (\alpha\llto\alpha) \fli (\alpha\llto\alpha) \fli
\pa ((\alpha\llto\alpha)^{2}\lltensor (\BOO^k \lltensor \BOO^2)),
$$
which serves as a type for the configurations of the ATM.

	We will proceed in the same way as Gaboardi et al.:
\begin{itemize}
\item show that all polynomials can be represented in the system;
\item define a function Step which answers recursively if a configuration will be accepted or not by the ATM, it will be given the type:\\
$(\mathbf{Conf}\fm \BOO^2) \fm (\mathbf{Conf}\fm \BOO^2)$;
\item define a term which decides if a given configuration is accepted, by iterating Step a polynomial number of times.
\end{itemize}

We have the following \DLALB-terms:
\begin{itemize}
\item $\mathsf{trans}_1\ (resp.\ \mathsf{trans}_2) : \mathbf{Conf}\fm \mathbf{Conf}$ for one-step of the first (resp. the second) function of transition of the ATM (similar to $\mathsf{trans}$ in \cite{Bai});
\item $\mathsf{Kind} : \mathbf{Conf}\fm \BOO^2$ for the projection from a configuration to its characteristic state;
\item $\mathsf{P} : \NNN\fm \pa^{2d}\NNN$ for the polynomial $n \mapsto\ n^{2^d}$ (same as $\mathsf{P}$ in \cite{Bai}).
\end{itemize}

	The term $\mathsf{Step}$ (of type $(\mathbf{Conf}\fm \BOO^2) \fm (\mathbf{Conf}\fm \BOO^2)$) is defined in a way analogous to $\mathsf{Step}$ in \cite{Gabo}:
\begin{itemize}
\item $\mathsf{Term_3} = if\ \pi_2\ (h\ (\mathsf{trans}_1\ c))\ then\ F\lltensor (\pi_2\ (h\ (\mathsf{trans}_2\ c)))\ else\ R$;
\item $\mathsf{Term_2} = if\ \pi_2\ (h\ (\mathsf{trans}_1\ c))\ then\ A\ else\ F\lltensor (\pi_2\ (h\ (\mathsf{trans}_2\ c)))$;
\item $\mathsf{Term_1} = if\ \pi_2\ (\mathsf{Kind}\ c)\ then\ \mathsf{Term_2}\ else\ \mathsf{Term_3}$;
\item $\mathsf{Step} = \la h. \la c. if\ \pi_1\ (\mathsf{Kind}\ c)\ then\ \mathsf{Term_1}\ else\ \mathsf{Kind}\ c$.
\end{itemize}

$\mathsf{Step}$ term operation:
\begin{itemize}
\item If $\mathsf{Step}$ receives as argument a configuration $c$ and a function of characterization $h$ of type $\mathbf{Conf}\fm \BOO^2$ such that $h\ (\mathsf{trans}_1\ c)$ (resp. $h\ (\mathsf{trans}_2\ c$) returns $A$ if $\mathsf{trans}_1\ c$ is accepted by the ATM and $R$ if it is rejected (resp. $A$ if $\mathsf{trans}_2\ c$ is accepted by the ATM and $R$ if it is rejected) then $\mathsf{Step}\ h\ c$ returns $A$ if $c$ is accepted by the ATM and $R$ if it is rejected;
\item $\mathsf{Term_1}$ represents the case where the characteristic state of $c$ is neither Accepting nor Rejecting (in which cases it is sufficient to return $\mathsf{Kind}\ c$);
\item $\mathsf{Term_2}$ (resp. $\mathsf{Term_3}$) represents the case where the characteristic state of $c$ is Existential (resp. Universal).
\end{itemize}

We also have the following \DLALB-terms:
\begin{itemize}
\item $\mathsf{init} : \WWW\fm \mathbf{Conf}$ for initialization (similar to $\mathsf{init}$ in \cite{Bai});
\item $\mathsf{length} : \WWW\fm \NNN$ for the length map (similar to $\mathsf{length}$ in \cite{Bai});
\item $\mathsf{coer} : \WWW\fm \pa^{2d} \WWW$ for an identity function (usefull for the typing and similar to $\mathsf{coer}$ in \cite{Bai}).
\end{itemize}

Finally we have $\mathsf{M} : \WWW\fm \pa^{2d+2}\BOO$ which is the term representing the ATM $\mathcal{M}$:
$\mathsf{M} = \la w. (\pi_2\ (\mathsf{P}\ (\mathsf{length}\ w)\ \mathsf{Step}\ \mathsf{Kind}\ (\mathsf{init}\ (\mathsf{coer}\ w))))$.

$\mathsf{M}$ term operation:
\begin{itemize}
\item $\mathsf{init}\ (\mathsf{coer}\ w))$ is a term which represents the initial configuration of the ATM.
\item $\mathsf{Step}$ calls itself recursively $n^{2^d}$ times (with $n$, represented by $\mathsf{length}\ w$, the length of the word $w$) -thanks to the term $\mathsf{P}\ (\mathsf{length}\ w)$- so that it calls $\mathsf{Kind}$ only on configurations which have a characteristic state Accepting or Rejecting. Thus the term $\mathsf{P}\ (\mathsf{length}\ w)\ \mathsf{Step}\ \mathsf{Kind}\ (\mathsf{init}\ (\mathsf{coer}\ w))$ returns $A$ if $w$ is accepted by the ATM and $R$ if it is rejected.
\item Therefore $\pi_2\ (\mathsf{P}\ (\mathsf{length}\ w)\ \mathsf{Step}\ \mathsf{Kind}\ (\mathsf{init}\ (\mathsf{coer}\ w)))$ returns $T$ (true) if $w$ is accepted by the ATM and $F$ (false) if it is rejected. Thus $\mathsf{M}$ represents $\mathcal{M}$.
\end{itemize}

 \begin{figure*}

  \begin{center}
\fbox{
\begin{tabular}{c}

Let $\De\ =\  h : \mathbf{Conf}\fm \BOO^2, c : \mathbf{Conf}$\\
\\

\tiny{

\infer[\mbox{($\fm$ e)}]{; \De \vdash \pi_2\ (h\ (\mathsf{trans}_1\ c)) : \BOO}
                           {\infer[\mbox{($\forall$ e)}]{; \vdash \pi_2 : \BOO^2 \fm \BOO}{\infer[\mbox{($\forall$ e)}]{; \vdash \pi_2 : \forall \beta. (\BOO\lltensor\beta) \fm \beta}{; \vdash \pi_2 : \forall \alpha. \forall \beta. (\alpha\lltensor\beta) \fm \beta}} &
                            \infer[\mbox{($\fm$ e)}]{; \De \vdash (h\ (\mathsf{trans}_1\ c)) : \BOO^2}
                                                        {\infer[\mbox{(Id)}]{; h : \mathbf{Conf}\fm \BOO^2 \vdash h : \mathbf{Conf}\fm \BOO^2}{} & \infer[\mbox{($\fm$ e)}]{; c : \mathbf{Conf} \vdash \mathsf{trans}_1\ c : \mathbf{Conf}}{; \vdash \mathsf{trans}_1 : \mathbf{Conf}\fm \mathbf{Conf} & \infer[\mbox{(Id)}]{; c : \mathbf{Conf} \vdash c : \mathbf{Conf}}{}}}
                           }
             }\\
\\

\small{

\infer[\mbox{($\lltensor$ i)}]{; \De \vdash F\lltensor (\pi_2\ (h\ (\mathsf{trans}_2\ c))): \BOO^2}
                           {\infer[\mbox{($B_0$ i)}]{; \vdash F : \BOO}{} &
                            ; \De \vdash \pi_2\ (h\ (\mathsf{trans}_2\ c)) : \BOO
                           }
             }\\
\\

\small{

\infer[\mbox{($B$ e)}]{; \De \vdash if\ \pi_2\ (h\ (\mathsf{trans}_1\ c))\ then\ A\ else\ F\lltensor (\pi_2\ (h\ (\mathsf{trans}_2\ c))) : \BOO^2}
  {; \De \vdash \pi_2\ (h\ (\mathsf{trans}_1\ c)) : \BOO
   & ; \De \vdash A : \BOO^2
   & ; \De \vdash F\lltensor (\pi_2\ (h\ (\mathsf{trans}_2\ c))): \BOO^2
  }
             }

\end{tabular}
}
\end{center}
  \caption{Type derivation for the term $\mathsf{Term_2}$}\label{Term2type}
\end{figure*}

%%%%%%%%%%%%%%%%%%%%%%%%%%%%%%%%%%%%%%%%%%%%%%%%%%%%%%%%%%%%%%%%%%%%%%%%%%%%%%%%%%%%%%%%%%%%%%%%%%%%%%%%%%%%%%%%%%%%%%%%%%%%%%%%%%%%%%%%%%%%%%%%%%%%%%%%%%
%%%%%%%%%%%%%%%%%%%%%%%%%%%%%%%%%%%%%%%%%%%%%%%%%%%%%%%%%%%%%%%%%%%%%%%%%%%%%%%%%%%%%%%%%%%%%%%%%%%%%%%%%%%%%%%%%%%%%%%%%%%%%%%%%%%%%%%%%%%%%%%%%%%%%%%%%%
%%%%%%%%%%%%%%%%%%%%%%%%%%%%%%%%%%%%%%%%%%%%%%%%%%%%%%%%%%%%%%%%%%%%%%%%%%%%%%%%%%%%%%%%%%%%%%%%%%%%%%%%%%%%%%%%%%%%%%%%%%%%%%%%%%%%%%%%%%%%%%%%%%%%%%%%%%
%%%%%%%%%%%%%%%%%%%%%%%%%%%%%%%%%%%%%%%%%%%%%%%%%%%%%%%%%%%%%%%%%%%%%%%%%%%%%%%%%%%%%%%%%%%%%%%%%%%%%%%%%%%%%%%%%%%%%%%%%%%%%%%%%%%%%%%%%%%%%%%%%%%%%%%%%%
%%%%%%%%%%%%%%%%%%%%%%%%%%%%%%%%%%%%%%%%%%%%%%%%%%%%%%%%%%%%%%%%%%%%%%%%%%%%%%%%%%%%%%%%%%%%%%%%%%%%%%%%%%%%%%%%%%%%%%%%%%%%%%%%%%%%%%%%%%%%%%%%%%%%%%%%%%
%%%%%%%%%%%%%%%%%%%%%%%%%%%%%%%%%%%%%%%%%%%%%%%%%%%%%%%%%%%%%%%%%%%%%%%%%%%%%%%%%%%%%%%%%%%%%%%%%%%%%%%%%%%%%%%%%%%%%%%%%%%%%%%%%%%%%%%%%%%%%%%%%%%%%%%%%%
%%%%%%%%%%%%%%%%%%%%%%%%%%%%%%%%%%%%%%%%%%%%%%%%%%%%%%%%%%%%%%%%%%%%%%%%%%%%%%%%%%%%%%%%%%%%%%%%%%%%%%%%%%%%%%%%%%%%%%%%%%%%%%%%%%%%%%%%%%%%%%%%%%%%%%%%%%
%%%%%%%%%%%%%%%%%%%%%%%%%%%%%%%%%%%%%%%%%%%%%%%%%%%%%%%%%%%%%%%%%%%%%%%%%%%%%%%%%%%%%%%%%%%%%%%%%%%%%%%%%%%%%%%%%%%%%%%%%%%%%%%%%%%%%%%%%%%%%%%%%%%%%%%%%%
%%%%%%%%%%%%%%%%%%%%%%%%%%%%%%%%%%%%%%%%%%%%%%%%%%%%%%%%%%%%%%%%%%%%%%%%%%%%%%%%%%%%%%%%%%%%%%%%%%%%%%%%%%%%%%%%%%%%%%%%%%%%%%%%%%%%%%%%%%%%%%%%%%%%%%%%%%

\section{Conclusion and perspectives}

We have presented a polymorphic type system for lambda calculus with booleans which guarantees that all well-typed terms are representing APTIME predicates and that all predicates of APTIME are represented by well-typed terms. Thus this system is characterizing PSPACE (given that PSPACE = APTIME).\\
Otherwise, if we were to consider terms of type $\WWW \fli \pa^n \WWW$ instead of terms of type $\WWW \fli \pa^n \BOO$ we believe that we would obtain a characterization of FPSPACE without changing the type assignment system and with the same data type in input and output (which is a property not shared by \STAB).\\
Now, it would be interesting to see if system \DLALB\ could be modified in order to characterize the polynomial hierarchy (PH). We think that such study would be facilitated by the use of APTIME Abstract Machine in the Soundness part of the proof of this paper. Thus this proof could be reused to prove the PH soundness of the modified system of \DLALB.

\bibliographystyle{eptcs} \bibliography{dlalbool}

\end{document}